\documentclass[11pt]{article}
\usepackage{fullpage}
\usepackage{amsmath}
\usepackage{algorithm}
\usepackage{amsthm}
\usepackage{amsfonts}
\usepackage{multirow}
\usepackage{amssymb}
\usepackage{fullpage}
\usepackage{graphics}
\usepackage{epstopdf}
\usepackage{epsfig}
\usepackage{subfigure}
\usepackage[usenames]{color}
\usepackage[margin=1in]{geometry}
\usepackage{setspace}

\theoremstyle{plain}
\newtheorem{theorem}{Theorem}
\newtheorem{claim}{Claim}[section]
\newtheorem{lemma}{Lemma}[section]
\newtheorem{proposition}{Proposition}
\newtheorem{observation}{Observation}[section]
\theoremstyle{definition}
\newtheorem{definition}{Definition}
\newtheorem{corollary}{Corollary}
\newtheorem{example}{Example}
\newtheorem{remark}{Remark}

\newlength{\bibitemsep}
\newlength{\bibparskip}
\let\oldthebibliography\thebibliography
\renewcommand\thebibliography[1]{%
  \oldthebibliography{#1}%
  \setlength{\parskip}{\bibitemsep}%
  \setlength{\itemsep}{\bibparskip}%
}

\title{Best Cost-Sharing Rule Design for Selfish Bin Packing}

\author{Changjun Wang\footnote{Academy of Mathematics and Systems Science, Chinese Academy of Sciences, Beijing, China; \textsf{wcj@amss.ac.cn}} \and Guochuan Zhang\footnote{College of Computer Science, Zhejiang University, Hangzhou, China; \textsf{zgc@zju.edu.cn}} }
\date {}

\begin{document}

\maketitle
\begin{abstract}
In selfish bin packing, each item is regarded as a selfish player, who aims to minimize the cost-share by choosing a bin it can fit in. To have a least number of bins used, cost-sharing rules play an important role. The currently best known cost sharing rule has a \emph{price of anarchy} ($PoA$) larger than 1.45, while a general lower bound 4/3 on $PoA$ applies to any cost-sharing rule under which no items have the incentive to move unilaterally to an empty bin. In this paper, we propose a novel and simple rule with a $PoA$ matching the lower bound of $4/3$, thus completely resolving this game. The new rule always admits a Nash equilibrium and its \emph{price of stability} ($PoS$) is one. Furthermore, the well-known bin packing algorithm $BFD$ (Best-Fit Decreasing) is shown to achieve a strong equilibrium, implying that a stable packing with an asymptotic approximation ratio of $11/9$ can be produced in polynomial time. As an extension of the designing framework, we further study a variant of the selfish scheduling game, and design a best coordination mechanism achieving $PoS=1$ and $PoA=4/3$ as well.
\end{abstract}
\section{Introduction}
Bin packing, one of the fundamental combinatorial optimization models, is a sort of resource-sharing problems. Given a sufficient number of identical resource slots (bins), assign a set of resource requests (items) to a minimum number of slots such that the total amount of requests in each assigned slot is bounded by the given capacity (usually normalized as one). As one of the very first NP-hard problems, bin packing has been extensively studied, while new results are still coming. 
The most common measure {\em asymptotic approximation ratio} evaluates the worst-case ratio between the cost of an approximation algorithm and the optimal cost in the sense that the optimal cost is arbitrarily large. Along this line, the well-known bin packing algorithm $BFD$ (Best Fit Decreasing), which will be formally introduced in Section \ref{se:BFD}, admits an asymptotic approximation ratio $11/9$~\cite{J74}. Namely,
$$
BFD(L)\le \frac{11}{9} OPT(L)+c
$$
is valid for any bin packing instance $L$, where $c$ is a constant, $BFD(L)$ is the number of bins used by $BFD$, $OPT(L)$ is the optimum number of bins, on $L$. 
An APTAS (asymptotic polynomial time approximation scheme) was provided by Fernandez de la Vega and Lueker \cite{de81}, which admits an asymptotic approximation ratio arbitrarily close to one.

Improving the approximation on bin packing has not yet stopped. There is no evidence for bin packing to rule out a polynomial time algorithm which uses at most one more bin than an optimum packing, assuming $P\not=NP$. There have been a lot of efforts in reducing the number of extra bins used, among which the best known algorithm uses $O(\log OPT(L))$ more bins than an optimal packing, owing to the elegant work by Hoberg and Rothvoss \cite{H17}.

In the last two decades, game settings were introduced to classical combinatorial optimization problems. In selfish bin packing, each bin has a unit cost, which is shared by all items using the same bin. Each item is handled by an agent (or simply regards an item as a selfish player), who aims to pay the least cost-share. Instead of packing algorithms, cost-sharing rules play an important role in this game, based on which the agents play selfishly to choose a bin. A packing solution is a Nash equilibrium if no items have an incentive to move unilaterally to a different bin with enough space. We also say such a packing is {\em stable}. To measure the efficiency of such a bin packing game, we still consider the number of bins used as the social objective. We are interested in how much loss a stable packing may have in comparison with an optimal packing in the classical setting without game issues (called social optimum). To this end we use $PoA$ (Price of Anarchy) \cite{K99,R02} and $PoS$ (Price of Stability) \cite{A08} to evaluate the game, which are defined as the worst case asymptotic ratio between the number of bins used in a worst (for $PoA$) and a best (for $PoS$) Nash equilibrium and the social optimum, respectively. For a selfish bin packing game $G$, let $OPT(G)$ be the social optimum, and $NE(G)$ be the set of Nash equilibria. For a stable packing $\pi\in NE(G)$, denote by $n(\pi)$ the number of bins used in $\pi$. Then PoA and PoS for the bin packing game are formally defined as \cite{E11}

$$PoA=\limsup_{OPT(G)\to\infty}\sup_{G}\max_{\pi\in NE(G)} \frac{n(\pi)}{OPT(G)}$$
and
$$PoS=\limsup_{OPT(G)\to\infty}\sup_{G}\min_{\pi\in NE(G)} \frac{n(\pi)}{OPT(G)}.$$
\subsection{Related Work}
The selfish bin packing game was initiated by Bil\'o \cite{B06}, who analyzed a natural proportional cost-sharing rule, under which an item pays the cost proportional to its size. It was proved that there exists a Nash equilibrium achieving the social optimum, implying that $PoS=1$, while $PoA$ falls in the interval $[1.6,1.67]$. Later two groups were focused on narrowing the interval. Epstein and Kleiman \cite{E11}, and Yu and Zhang \cite{Y08} independently got a lower bound of around 1.6416, while Epstein and Kleiman \cite{E11} reached a better upper bound of 1.6428. The gap is very small but still there. Epstein et al. \cite{E16} further considered the case that item sizes are restricted and a parametric bound on $PoA$ was derived.

Another natural rule, called the equally sharing rule, simply asks the items in the same bin to pay the same cost regardless of their sizes. Ma et al. \cite{M13} showed that under this rule $PoA$ is in $[1.6901,1.7]$, which is even worse than the proportional rule. Very recently, D\'osa and Epstein \cite{D20} reduced an upper bound of $PoA$ down to below 1.7.

Both the two natural rules do not seem promising as the inefficiency ($PoA$) is too big. It is thus well motivated to design better cost-sharing rules to improve the efficiency of Nash equilibria. There have appeared a number of interesting results in recent years. Nong et al. \cite{N18} proposed a rule, where the payoff of an item is a function of its own size and the largest item size in its bin. They obtained a better $PoA$ of 1.5. 
Later, Chen et al. \cite{C21}  designed another cost-sharing mechanism with PoA being between 1.47407 and 1.4748. Zhang and Zhang \cite{Z20} kept the proportional rule for those bins without large items (whose size is larger than 1/2), while offering a discount function for the items staying with a large item. They derived a better cost-sharing rule with $PoA$ at most $22/15\approx 1.467$, the same upper bound for the proportional rule applying to bin packing games without large items \cite{E16}. D\'osa et al. \cite{Dosa2019} dealt with the problem in a different way. They applied the proportional rule to the item weights instead of the item sizes. By carefully defining the item weights, they proved the $PoA$ is within $1.4528$ and $1.4545$ \cite{Dosa2019}.

Recall that the classical bin packing problem admits an APTAS, while the currently best-known upper bound on $PoA$ for selfish bin packing is 1.4545. It leaves a big gap. Does there exist a cost-sharing rule, under which the $PoA$ can be arbitrarily close to one? Unfortunately the answer is negative for a wide class of the bin packing games. D\'osa et al. \cite{Dosa2019} presented an instance showing that the $PoA$ is at least $4/3$ if the cost-sharing rule does not encourage a packed item to move to an empty bin. Such an assumption is quite natural, as long as no items will pay a cost more than one. There are also quite a few results in the literature considering different variants of selfish bin packing, such as selfish bin covering \cite{C11}, selfish vector packing \cite{E21}.

There is also an extensive study in the literature that shares a similar favor as our work: how to  design  prior game rules to minimize the inefficiency induced by selfish players. The line of research on coordination mechanisms was introduced by Christodoulou et al. \cite{CKN09}, and subsequently studied by Immorlica et al. \cite{ILMS09}, Caragiannis \cite{C13}, Kollias \cite{K13}, Azar et al. \cite{AFJMS15}, and Cole et al. \cite{CCGMO15}. Most work on coordination mechanisms concerns scheduling games in which $n$ players assign a respective job to one of $m$ machines with a goal to minimize its own completion time, and seeks how the price of anarchy varies with the choice of local machine scheduling policies (i.e., the order in which to process jobs assigned to the same machine). Another mainstream of such studies is on the design of cost-sharing protocols for resource selection games~\cite{CRV10,vH13,Hv14,GKR16,HHSS21}, where selfish players choose a set of resources to use, and the cost of each resource is a function of the set of players choosing the resource. In these works, the model settings are somewhat different from the bin packing game and their allowable cost-sharing methods are more flexible. Therefore the proved bounds and corresponding approaches therein do not apply to the bin packing game. Moreover, the design of scheduling policy and the design of cost-sharing protocols are believed to be fundamentally different \cite{vH13,Hv14}. As far as we know, no work has shown that these two rule-designs are closely related. Interestingly, however, in this paper, we will show that the best cost-sharing rule in some specific resource selection games and the best scheduling policy in some specific scheduling games are essentially equivalent.

\subsection{Our Contributions} In this paper, we revisit selfish bin packing by proposing a novel and simple cost-sharing scheme, which is based on the local item sizes in a bin. Basically, we introduce a threshold parameter $0<\Lambda\le 1$, and stack up the items in decreasing order of their sizes in a bin. If the total size is at least $\Lambda$, then the items pay the cost based on a cost density function ranging in $[0, \Lambda]$; otherwise, the bottom item (the largest one) has to pay some extra cost. The cost-sharing rules satisfy many nice and natural properties. 
By setting $\Lambda=\frac34$, we show that there always exists an optimum packing which is a Nash equilibrium, implying that $PoS=1$. Then we prove that the $PoA$ is at most $\frac43$, matching the lower bound, and thus resolving the wide class of games. Finally, we turn to the classical $BFD$ algorithm which works perfectly under our cost-sharing rule by setting $\Lambda=\frac23$. $BFD$ always produces a strong Nash equilibrium, implying that we can output in $O(n\log n)$ time a stable packing whose asymptotic approximation ratio is $\frac{11}{9}$ ($<\frac43$). It is also currently the best bound of a Nash equilibrium carried out in polynomial time.

Furthermore, as an extension, we study a variant of the selfish scheduling game, and design a best coordination mechanism with $PoS=1,\ PoA=\frac43$. We show that the best cost-sharing rule in the selfish bin packing is essentially equivalent to the best scheduling policy of this scheduling game.

\medskip

\noindent\textbf{Organization of the Paper.} In the next section, we formally introduce the bin packing game model and the new cost-sharing rules together with several observations. Section \ref{se:analysis} presents the analysis of the rule by showing the main results on $PoS$ and $PoA$. The analysis of algorithm $BFD$ is provided in Section \ref{se:BFD}. In Section \ref{se:extend}, we extend the results to a variant of the selfish scheduling game. Finally, we conclude the paper in Section \ref{se:con}.

\section{Game Model and A New Cost-Sharing Scheme}
In this section, we formally introduce the game model of selfish bin packing and some definitions. Then we propose our new cost-sharing rules that satisfies many nice properties.
\subsection{Definitions}
Let $N=\{1,2,\ldots,n\}$ be a set of items, where each item $i\in N$ has a size $s_i\in(0,1]$ and is controlled by a selfish agent. For simplicity, in the following we will directly view the items  as game players. There are plenty of  unit-capacity bins being available to  pack the items, but each with an unit opening cost. For a set of items $I\subseteq N$ packed into a  bin
$B$, we say that $B$ is valid or a valid bin, if $s(I):=\sum_{i\in I}s_i\le 1$. A packing $\pi$ is defined to be an allocation of the items into some  used bins $B_1,\ldots, B_m$  such that $N = B_1 \cup B_2\cdots \cup B_m$ and each bin $B_k$ is valid, where $B_k$ ($1\le k\le m$) denotes a bin as well as the set of items allocated in that bin (slightly abusing notation). For the items packed into one bin, they need to share the bin's opening cost according to a \emph{pre-given} cost-sharing rule $\Xi$. We use $c_i(\pi)\geq 0$ to represent the cost-share of item $i$ under packing $\pi$. Each item aims to minimize its own cost-share. The strategy $b_i$ of each item $i\in N$ is to select a valid bin to stay in. Then a feasible strategy profile $\mathbf b$ is  equivalent a unique packing $\pi=(B_1,\ldots,B_m)$ with nonempty bins $B_k$, i.e., $\mathbf b$ and $\pi$ can be transformed by  $B_k=\{i\in N: b_i=B_k\}$ for  $k=1,\ldots,m$,  or by  $b_i=B_k$ for every $i\in B_k$. 
The social objective is to minimize the  number of used bins to pack all the items. 

\begin{definition}\label{NE}
A strategy profile $\mathbf b$ (or a packing $\pi$) is called a \emph{Nash equilibrium} (\emph{NE}, for short) if no item can reduce its cost by unilaterally changing its strategy, i.e., $c_i(\mathbf b)\le c_i(b_i', \mathbf{b}_{-i})$,  $\forall i\in N$.
\end{definition}

Designing cost sharing rules to improve the efficiency of equilibrium (i.e., minimize PoA or PoS) for different resource selection games has been widely studied in literature~\cite{CRV10,vH13,Hv14,GKR16,Z20,HHSS21}. In these studies,  many desirable properties have been proposed to restrict the design space of feasible cost sharing rules. Below we define them and some new natural ones in this context.

\begin{definition}[{\sc Properties of Cost Sharing Rules}]\label{axiom}
A cost sharing rule $\Xi$ is:
\begin{enumerate}
\item \emph{stable}, if it induces only games that admit at least one NE.
\item \emph{budget-balanced}, if for every  bin $B$ of any packing $\pi$, $\sum_{i\in B}c_i(\pi)=1$.
\item \emph{local}, if for any two different packings $\pi\neq\pi'$ that both contain a  bin $B$ with the same set of packed items, and every item $i\in B$, it holds $c_i(\pi)=c_i(\pi')$.
\item \emph{monotone},  if for any two bins $B,B'$ satisfying $B\subseteq B'$ (suppose w.l.o.g. $B$ is a bin of  packing $\pi$ and $B'$ is with $\pi'$) and every item $i\in B$, $c_i(\pi)\geq c_i(\pi')$.
\item \emph{fair},  if for any two items $i,j$ that are packed in the same bin of any packing $\pi$, it holds $(s_i-s_j)\cdot\big(c_i(\pi)- c_j(\pi)\big)\geq 0$.
\item \emph{polynomial time computable}, if for every packing $\pi$, the cost-shares $c_i(\pi)$ of all $i\in N$ can be computed in polynomial time of the input size.
\end{enumerate}
\end{definition}

Now we discuss the above properties in more detail. Note that  the NEs in Definition \ref{NE} are in fact  \emph{pure strategy} Nash equilibria. Thus a stable cost sharing rule ensures the existence of a pure strategy Nash equilibrium  in the induced game of selfish bin packing.  Budget-balance is a natural, straightforward requirement for a cost sharing rule in the economics literature. Locality means that for any bin, the cost-shares of the items in this bin depend only on the bin itself (i.e., the set of items in this bin),  and disregard the information of other bins. A local cost sharing rule ensures that the opening costs can be distributed in a decentralized and local manner, that is, the bin does not have to know any cost-shares of other items or the allocation for other bins. Monotonicity stating that moving  new items into a bin should not increase the cost-shares of previous items of this bin  is  a natural property in selfish bin packing as the  cost of a bin is always a constant. Since the capacity of each bin is limited, in a fair cost-sharing rule, for the items packed in the same bin,  an item with a larger size should not pay less than a smaller item. Polynomial-time computability of the rules is crucial to make the games run efficiently.

\subsection{A new cost-sharing scheme}
In this subsection, we will propose our new cost-sharing scheme. First recall that under the classical propositional cost-sharing rule, the items are incentivized to move into feasible bins with more total loads. Differently from that, the rough idea of our new rules is to motivate each item $i\in N$ to choose a feasible bin with more total loads of larger items (than $i$'s size). Even though this motivation seems no more effective (at getting the bins more loaded) than the propositional rule, we will finally show that this kind of strategic behaviors will result in the best possible NE efficiency. However, to realize this kind of incentive and make the nice properties in Definition \ref{axiom} satisfied, the new cost-sharing scheme needs to be well designed. 

\paragraph{Local-Size-Based Cost-sharing Scheme.}
The new framework of cost-sharing rules is parameterized with a constant $\Lambda\in(0,1]$ and  a ``cost density" function $f(x)$. Namely, a constant $\Lambda$ and a function $f(x)$ will uniquely determine one cost-sharing rule. Given a constant $\Lambda$, the function  $f(x)$ with $x\in[0,1]$ is a non-negative function satisfying: i).  $f(x)$ is decreasing in $[0,\Lambda)$ with $\int_0^{\Lambda}f(x)dx=1$; and ii). $f(x)=0$ for $x\in [\Lambda,1]$. A simple selection of $f(x)$ is letting $f(x)=f_o(x)$, where $f_o(x)$ is a linear function,
\[f_o(x):=\begin{cases} \frac{2}{\Lambda}-\frac{2x}{\Lambda^2}, \   \text{ for } 0\le x<\Lambda; \\ 0, \ \ \text{ for } \Lambda\leq x\leq 1.
\end{cases}\]

The local-size-based (LSB) cost-sharing rules  work  as follows:
\begin{itemize}
\item[1).]  Re-index the items in $N$ such that $s_1\ge s_2\ge \ldots\ge s_n$ and fix this index (order). If index $i<j$, we call item $i$ ``bigger'' than item $j$ or $j$ is ``smaller" than $i$ , denoted by item $i\succ$ item $j$. Note that a bigger item may have the same size as the smaller item.

\item[2).] Given a packing $\pi$, for each bin $B$ of $\pi$, suppose w.l.o.g. that $B=\{i_1,i_2,\ldots,i_l\}$ and $i_1\succ i_2\succ\cdots\succ i_l$. Denote by $S_B^h:=\sum_{j=1}^hs_{i_j}$ for $h\leq l$ the total size of the biggest  $h$ items in bin $B$. Note that $S_B^l=s(B)$. Then the cost share of item $i_h$ for $h=1,2,\ldots,l$ is
\[ c_{i_h}(\pi)=\begin{cases} \int_0^{S_B^1}f(x)dx+\big(1-\int_0^{s(B)}f(x)dx\big),  \  \text{  for  } h=1;\\ \int_{S_B^{h-1}}^{S_B^{h}}f(x)dx,  \ \hspace{5mm} \text{  for  } h>1.
\end{cases}
\]
\end{itemize}

\begin{figure}[H]
	\centering
	\includegraphics[width=9cm]{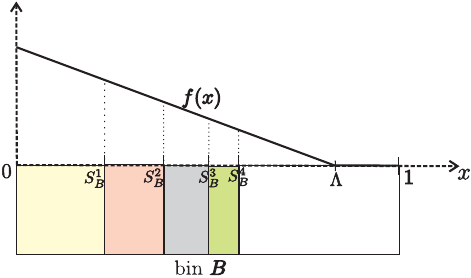}
	\caption{Cost sharing illustration}
		\label{CostSharing}
\end{figure}

In a bin, we call the biggest item (the item with smallest index) the bottom item of the bin.
In the above rules, except for the bottom item, an item's cost depends on the total size of bigger items and its own size, disregarding the sizes of smaller items. For the bottom item, if the total load of the bin $s(B)$ is less than $\Lambda$, then it should also pay for the extra remaining  cost $\int_{s(B)}^{\Lambda}f(x)dx$. Apparently, every LSB cost-sharing rule is budget-balanced and local.

\begin{observation}\label{axiom23}
The LSB cost-sharing rules are budget-balanced and local.
\end{observation}
In addition,  we have the following obvious observation about the LSB cost sharing scheme.
\begin{observation}\label{axiom45}
The LSB cost-sharing rules are  fair and monotone.
\end{observation}
\begin{proof}
By noting that the density function $f(x)$ is monotone non-increasing,  it is easy to see that the LSB cost-sharing rules are  fair. To see the monotonicity, pick  any two bins $B,B'$ satisfying $B\subsetneq B'$.  Suppose w.l.o.g. that $B$ (resp.\ $B'$) is a bin of  packing $\pi$ (resp.\ $\pi'$), and $j\in B'\setminus B$. Then for any item $i\in B$, if $i$ is bigger than $j$ but  not the bottom item of $B$,  then the LSB cost-sharing scheme guarantees that $i$'s cost share remains the same, $c_i(\pi)=c_i(\pi')$; otherwise, $i$ is either a bottom item of $B$ or smaller than $j$, then we have $c_i(\pi)\geq c_i(\pi')$.
\end{proof}

Generally,  the LSB rules may  not be  stable or polynomial time computable. But through selecting some proper $\Lambda$ and  $f(x)$, we can also guarantee that the corresponding LSB cost-sharing rule is  stable and polynomial time computable.
\begin{theorem}\label{good}
Let constant $\Lambda\leq3/4$ and $f(x)=f_o(x)$.  Then this LSB cost-sharing rule satisfies all the proposed properties in Definition \ref{axiom}.
\end{theorem}
\begin{proof}
Having observations \ref{axiom23} and \ref{axiom45}, we  only need to show that this cost-sharing rule is stable and polynomial time computable. The stability of the cost-sharing rule will be proved in Subsection \ref{existence}. Since $f_o(x)$ is a linear function and its integral can be given explicitly,  then  for any packing $\pi$, the time of computing all the cost shares is apparently polynomial in the instance's input.
\end{proof}

In the remaining part of the paper, when we mention selfish bin packing, if not otherwise specified, it is always assumed that the games are induced by the LSB cost-sharing scheme. As we mainly care about the NE packings and their performances in selfish bin packing,  below we will see that the density function $f(x)$ is not critical in the design of our LSB cost-sharing rule.
\begin{proposition}\label{sameNE}
Bin packing games with the same $\Lambda$ but different cost density functions have exactly the same set of NE packings.
\end{proposition}
\begin{proof}
Note that in these games,  a packing is an NE  if and only if no item of this packing can reduce its  cost by unilaterally moving into another feasible bin.  Pick any item $i\in N$ in any packing $\pi$.  With the LSB cost-sharing rule, if $i$ has incentive to move, then it must be the case that, either in the bin $B$ where $i$ stays, the total size of items that are bigger than $i$ is less than $\Lambda$ (otherwise $i$'s cost would be $0$) and there exists another feasible bin (for $i$ to move in) with larger total size of bigger (than $i$) items, or $i$ is a bottom item of its bin $B$ with $s(B)<\Lambda$ and there exists another feasible bin (for $i$) with total load larger than $s(B)-s_i$. 
Thus we can see, once the constant $\Lambda$ is given,  no matter what the exact density function $f(x)$ is, the NE packings are already uniquely determined.
\end{proof}

With Proposition \ref{sameNE} at hand,  if not otherwise specified in the following, we implicitly have $f(x)=f_o(x)$ and will only refer to bin packing games with certain constant $\Lambda$. As for games with different  $\Lambda$s, we show that their NE sets are also highly correlated.
\begin{proposition}\label{subsetNE}
For two different constants $\Lambda_1>\Lambda_2$, the NE set of bin packing game with $\Lambda_1$ is a subset of the NE set of game with $\Lambda_2$.
\end{proposition}
\begin{proof}
Suppose packing $\pi$ is an NE of the bin packing game with $\Lambda_1$. In the game with $\Lambda_2$, if $\pi$ is not an NE, then there must exist an item $i$ in some bin $B$ such that it can reduce its cost by unilaterally moving into another feasible bin $B'$. This means that either the total size of  items bigger than $i$ in $B$ is less than $\Lambda_2$ and $B'$ has a larger total size of bigger items, or $i$ is the bottom item of $B$ with $s(B)<\Lambda_2$ and $s(B')>s(B)-s_i$. Recall that $\Lambda_1>\Lambda_2$. We can deduce that in the bin packing game with $\Lambda_1$, under $\pi$, item $i$ can also reduce its cost by unilaterally moving to $B'$, contradicting the definition of $\pi$.
\end{proof}

\section{A Best LSB Cost-Sharing Rule}\label{34}\label{se:analysis}
In this section, we will focus on the bin packing games with constant $\Lambda=\frac34$ in our LSB cost-sharing rules, and show the LSB rule with $\Lambda=\frac34$ is the best possible in reducing  equilibria's inefficiency.

\subsection{Existence of Equilibria and $PoS=1$}\label{existence}

In the classical bin packing problems, an optimal solution is a packing with the minimum number of bins used. In the following, we will show that there always exists one optimal packing that being the NE of the game, which also implies its $PoS=1$.
\begin{theorem}\label{pos}
The bin packing game always has an NE that uses the minimum number of bins among all packings.
\end{theorem}
To prove Theorem \ref{pos}, we will seek optimal packings with some nice properties, among which we could figure out an NE. 
Given a packing $\pi$, for every item $i$ and the bin $B$ containing $i$, we refer to the total size of bigger items (than $i$) of $B$ as item $i$'s \emph{height level}, denoted by $l_i(\pi)$. Apparently in this game, every item would like to stay in a bin with its own height level as high as possible.  It would be relatively easy to ensure the existence of an optimal packing under which no item can unilaterally improve its height level. However, such packing is not necessarily an NE as some bottom item might still want to deviate when it stays in a bin with a total size less than $\Lambda$. This is because the bottom items do not only care about their height levels but also the total size of their bins. Thus the key challenge of proving Theorem \ref{pos} is how to guarantee the stability of  bottom items in the corresponding optimal packings.

Now let $m^*$ be the minimum number of used bins that can pack all the items of $N$, and $\mathcal {P}^*$ be the set of the optimal packings, each of which uses $m^*$ bins. For each packing $\pi\in\mathcal P^*$, we use $\{b^{\pi}_1,b^{\pi}_2\ldots, b^{\pi}_{m^*}\}$ with $b^{\pi}_1\succ\cdots\succ b^{\pi}_{m^*}$ to denote the set of the bottom items in the $m^*$ bins of $\pi$. For two packings $\pi,\pi'\in \mathcal P^*$, we say that $\pi$ lexicographically
dominates $\pi'$ with respect to the bottom items if either $b^{\pi}_1\succ b^{\pi'}_1$, or for some $k>1$ such that $b^{\pi}_j=b^{\pi'}_j$ for every $j<k$, and $b^{\pi}_k\succ b^{\pi'}_k$.  Broadly speaking, packing $\pi$ dominates $\pi'$ if its bottom items are lexicographically ``bigger" than $\pi'$'s bottom items. Let $\mathcal P_b^*\subseteq \mathcal P^*$ be the set of packings that can not be dominated by other packings in $\mathcal P^*$. Then apparently, $\mathcal P_b^*\neq\emptyset$ and all packings in $\mathcal P_b^*$ have the same set of bottom items, which is denoted by $\{b^*_1,b^*_2,\ldots, b^*_{m^*}\}$.

Then assume that the items in $N\setminus \{b^*_1,b^*_2,\ldots, b^*_{m^*}\}$ are $a_1\succ a_2\succ \cdots\succ a_{n-m^*}$.  For two packings $\pi,\pi'\in \mathcal P^*_b$,  we say that $\pi$ is lexicographically higher than $\pi'$ with respect to $\{a_1,\ldots,a_{n-m^*}\}$ if for some $t\le n-m^*$ the height levels satisfy: $l_{a_h}(\pi)=l_{a_h}(\pi')$ for every $h<t$ and $l_{a_{t}}(\pi)>l_{a_{t}}(\pi')$. In other words, for two packings $\pi,\pi'\in \mathcal P^*_b$, since they have the same set of bottom items, then $\pi$ is  higher than $\pi'$ if the height levels of $\{a_1,a_2,\ldots,a_{n-m^*}\}$ under $\pi$ are lexicographically higher than their height levels under $\pi'$.
Since $\mathcal P_b^*\neq\emptyset$, then there must exist a packing $\pi^*\in \mathcal P_b^*$ such that no other packing in $\mathcal P_b^*$ is lexicographically higher than it.  Let $B_1,\ldots,B_{m^*}$ be the bins that contain the bottom items $b^*_1,\ldots, b^*_{m^*}$ respectively. To ease the following discussions, we additionally require $\pi^*$ to satisfy the following condition: if there is $h<m^*$ such that $s_{b_h^*}=s_{b_{m*}^*}$, then the second biggest item of $B_{m^*}$ is not ``bigger" than the second biggest item of bin $B_h$ (including the cases of the second item being empty). This additional requirement can be easily realized by directly swapping all the non-bottom items of the corresponding two bins from a given $\pi^*$. 

Note that under any packing $\pi\in\mathcal P_b^*$, for every non-bottom item $i\in\{a_1,\ldots,a_{n-m^*}\}$, its cost share is $\int_{l_i(\pi)}^{l_i(\pi)+s_i}f(x)dx$. Recall that $f(x)$ is a non-increasing function. Thus under the selected packing $\pi^*$, by its definition, no item of $\{a_1, a_2, \ldots, a_{n-m^*}\}$ has the incentive to unilaterally deviate from its current bin since the single deviation cannot increase this item's height level (thus decrease its cost).
 If $\pi^*$ is an NE packing, then we are done. Otherwise, there must exist one unsatisfied bottom item that wants to move into another feasible bin to decrease its cost.  
 We first have the following claims.
\begin{claim}
If $\pi^*$ is not an NE packing, then the smallest bottom item $s_{b_{m^*}^*}\le 1/2$.
\end{claim}
\begin{proof}
Since  $\pi^*$ is not an NE, then at least one bottom item of some bin can unilaterally move into another bin. This means that a bin can accommodate some two bottom-items of $\{b^*_1,b^*_2,\ldots, b^*_{m^*}\}$, which implies  $s_{b_{m^*}^*}\le 1/2$.
\end{proof}
\begin{claim}\label{uniquelast}
If $\pi^*$ is not an NE packing, then the last bin $B_{m^*}$ contains only one item (i.e., $b^*_{m^*}$).
\end{claim}
\begin{proof}

Now suppose on the contrary that there are at least two items in $B_{m^*}$. Let $x$ be the second biggest item in $B_{m^*}$. We also use $x$ to represent its size if there is no confusion. It is easy to check that item $x$ can not fit into any other bin. Because if $x$ can fit into another bin (say $B_j$), recalling $b_j^* \succ b_{m^*}^*$ and the additional requirement of $\pi^*$,  it would imply that item $x$ would achieve a higher height level (than $l_{x}(\pi^*)=s_{b_{m^*}^*}$) after moving into $B_j$ and other bigger items' height levels could not be affected, contradicting the selection of $\pi^*$.

Since $\pi^*$ is not an NE, let $b_k^*$ be the bottom item that is willing to move from its current bin $B_k$ to another bin $B_h$, where $s_{b_k^*}+s(B_h)\leq 1$. Now we claim it must be the case that bin $B_h$ is exactly $B_{m^*}$. Because otherwise from the definition of item $x$ and the fact that $b_{m^*}^*$ is the smallest bottom item, we know that item $x$ can also fit into $B_h$, thereby leading to a contradiction to the above $x$'s fitness. Therefore,
 the  unsatisfied bottom item $b_k^*$ can only move from $B_k$ to $B_{m^*}$.

By the above definition of $b_k^*$ and the given LSB rule, we can deduce that $s(B_k)<3/4$, and $s(B_k)-s_{b_k^*}<s(B_{m^*})\le 1-s_{b_k^*}$. Since $x$ cannot fit into $B_k$, while $b_k^*$ ($\succ x$) fits into $B_{m^*}$, we have $x>1-s(B_k)>1/4$ and $s(B_{m^*})<s(B_k)<3/4$. Then we can deduce that $B_k$ contains at most two items that are bigger than item $x$.

Let $y$ be the second biggest item of $B_k$  (by slightly abusing notation, its size is also denoted by $y$). Then we have  $b_{m^*}^*\succ y$, because otherwise, by moving $b_k^*$ from $B_k$ to $B_{m^*}$, we would obtain another optimal packing that dominates all the packings in $\mathcal P_b^*$ w.r.t. the bottom items, contradicting the definition of $\mathcal P_b^*$. Define $I_k':=B_k\setminus\{b_k^*,y\}$, $I_{m^*}':=B_{m^*}\setminus \{b_{m^*}^*\}$.
Since $b_k^*$
can unilaterally move into $B_{m^*}$, thus  $$s_{b_k^*}+y+s(I_{m^*}')\leq s_{b_k^*}+s_{b_{m^*}^*}+s(I_{m^*}')=s_{b_k^*}+s(B_{m^*})\leq 1.$$ Meanwhile, $$s_{b_{m^*}^*}+s(I_k')<s(B_k)<3/4.$$
Therefore, by swapping $I_k'$ and $I_{m^*}'$, we will get a new optimal packing $\pi^{\circ}\in \mathcal P_b^*$ from $\pi^*$, with the two bins $B_k, B_{m^*}$ of $\pi^*$ replaced by another two  feasible bins $B_k'=\{b_k^*,y\} \cup I_{m^*}', B_{m^*}'=\{b_{m^*}^*\}\cup I_k'$.

Recall that item $x$ is bigger than any item in $I_k'$.  By noting $x\in B_k'$ in the new packing $\pi^{\circ}$, and combining with $b_k^*\succ b_{m^*}^*$ and the additional requirement of $\pi^*$, it is easy to see that packing $\pi^{\circ}$ is lexicographically higher than $\pi^*$, contradicting the definition of $\pi^*$. Now we can conclude that if $\pi^*$ is not an NE, then $b^*_m$ is the unique item in bin $B_{m^*}$.
\end{proof}

It's possible that $\pi^*$ is not an NE. So in the set $\mathcal P_b^*$,  we cannot guarantee the existence of an NE packing.  Fortunately, in that case, we would have a nice estimation of the last bin's size. Now we expand the set $\mathcal P_b^*$ to the set  $\mathcal  P_{b^-}^*\subseteq \mathcal P^*$, which consist of all the optimal packings that only have  $b^*_1,b^*_2,\ldots, b^*_{m^*-1}$ as the fixed bottom items (i.e, relax the possibility of the last bin's bottom item). Apparently, $\mathcal P_b^* \subseteq  \mathcal  P_{b^-}^*$. We will show that $\mathcal  P_{b^-}^*$ must contain an NE packing. Let $\hat\pi^*:=\hat B_1\cup \hat B_2\cup\cdots\cup\hat  B_{m^*}\in \mathcal  P_{b^-}^*$ be the packing that has the minimum total size of the last bin $\hat B_{m^*}$, where $\hat B_k$ has the item $b^*_k$ as the bottom item for $k=1,2,\ldots,m^*-1$. Besides, we refine $\hat\pi^*$ such that all the non-bottom items of the first $m^*-1$ bins (i.e., the items in $\hat B_1\cup \hat B_2\cup\cdots\cup\hat  B_{m^*-1}\setminus\{b^*_1,b^*_2,\ldots, b^*_{m^*-1}\}$) are in the lexicographically highest levels. That is, under $\hat\pi^*$, these non-bottom items have no incentive to unilaterally migrate to another bin of these $m^*-1$ bins.  Now we are ready to complete the proof of Theorem \ref{pos}.



\bigskip

\noindent\textbf{Proof of Theorem \ref{pos}}
\begin{proof}
We claim that either $\pi^*$ or $\hat\pi^*$ is an NE packing. Now suppose $\pi^*$ is not an NE.  Then from the definition of $\hat\pi^*$ and Claim \ref{uniquelast} we know, $s(\hat  B_{m^*})< s_{b_{m^*}^*}\le s_{b_{m^*-1}^*}\le \cdots\le s_{b_1^*}$. In this case, for all the non-bottom items in $\hat B_1\cup \hat B_2\cup\cdots\cup\hat B_{m^*-1}$, no one would like to unilaterally move into the last bin $\hat B_{m^*}$. This is because such migration would lower the item's height level from at least $s_{b_{m^*-1}^*}$ to $s(\hat  B_{m^*})$, thus increasing its cost share. Combining with the refinement of $\hat\pi^*$, we know that $\hat\pi^*$ is a stable packing for all the non-bottom items in $\hat B_1\cup \hat B_2\cup\cdots\cup\hat  B_{m^*-1}$. For any item in the last bin $\hat B_{m^*}$, no one can unilaterally move into the first $m^*-1$ bins since otherwise, the resulting new packing will have a smaller total size of the last bin than $\hat\pi^*$. This also implies that among the first $m^*-1$ bins, no bottom item of one bin can unilaterally fit into another bin. Besides, for any bottom item  $b_k^*$ with  $k\in\{1,2,\ldots, m^*-1\}$, if it can fit into the last bin $\hat B_{m^*}$, then we must have $s(\hat B_k)-s_{b_k^*}\geq s(\hat B_{m^*})$ because otherwise, moving  $b_k^*$ into $\hat B_{m^*}$ would result in a new packing with an even smaller size of the ``last bin".  So  under  $\hat\pi^*$, every bottom item of $\{b^*_1,b^*_2,\ldots, b^*_{m^*-1}\}$ is stable. Therefore, from the above analysis we can conclude that if $\pi^*$ is not an NE, then $\hat\pi^*$ must be an NE packing.
\end{proof}

Combining Proposition \ref{subsetNE} and Theorem \ref{pos}, we can also reach the following corollary for games with $\Lambda\le \frac34$.
\begin{corollary}
For every bin packing game  with $\Lambda\le \frac34$, its cost-sharing rule is stable and $PoS=1$.
\end{corollary}

\begin{remark}
From Proposition \ref{subsetNE} we know, increasing the constant $\Lambda$ in the LSB cost-sharing rule will make the corresponding NE set shrink, thus is helpful to reduce the game's $PoA$. It seems that we should set $\Lambda=1$ instead of $\Lambda=\frac34$.  However, for the bin packing games with $\Lambda=1$, we do not know whether an NE exists. In fact, for the bin packing games with $\Lambda>\frac34$,  the existence of NEs and their exact $PoS$ are still unknown and open. Fortunately, in the next subsection, we will show that the LSB cost-sharing rule with $\Lambda=\frac34$ is already the optimal rule to attain the best possible $PoA=4/3$.
\end{remark}

\subsection{Best possible $PoA$}\label{poa}

In this section, we prove that the $PoA$ of the bin packing game with $\Lambda=\frac34$ is at most $4/3$. From \cite{Dosa2019} we know, this bound is the best possible. 

Given an NE packing $\pi$ of items in $N$. Suppose without loss of generality that $\pi$ is composed of $m$ nonempty bins $B_1, B_2,\ldots, B_m$ with $s(B_1)\ge s(B_2)\ge \cdots\ge s(B_m)$. For each bin $B_k$, we still use $b_k$ to denote its bottom item, meanwhile, let $t_k$ be the smallest item in it.
\begin{lemma}\label{surplus}
For each bin $B_k$, if its smallest item pays (positive) cost, then the front bins are not feasible for any item of $B_k$ to move in, i.e., if $s(B_k)-s_{t_k}<\frac34$, then  $\forall j<k$, we have $s(B_j)+s_{t_k}>1$.
\end{lemma}
\begin{proof}
Suppose on the contrary that there exists  $j<k$ such that  $s(B_j)+s_{t_k}\leq 1$, which means that $t_k$ can fit into $B_j$.  Let $B_j'\subseteq B_j$ be the subset of items of $B_j$ that are bigger than $t_k$, i.e., $B_j':=\{i\in B_j: i\succ t_k\}$. Since $s(B_k)-s_{t_k}<\frac34$, by recalling the LSB cost-sharing rule and the fact that $\pi$ is an NE, then it must hold that $s(B_j')\le s(B_k)-s_{t_k}<\frac34$, because otherwise $t_k$ would be able to reduce its cost by unilaterally moving into $B_j$. At this case, we have $B_j\setminus B_j'\neq\emptyset$ since $s(B_j')\le s(B_k)-s_{t_k}<s(B_k)\le s(B_j)$. Let $i_j\in B_j\setminus B_j'$ be the biggest item 
of $B_j\setminus B_j'$. Then  $i_j$'s cost share under $\pi$ is $\int_{s(B_j')}^{s(B_j')+s_{i_j}}f(x)dx$.  Note that all items in $B_k$ are bigger than item $i_j$ since $t_k\succ i_j$ and  $s(B_k)+s_{i_j}\le s(B_j)+s_{t_k}\leq 1$. This means that under packing $\pi$, item $i_j$ can also unilaterally move into bin $B_k$ with a new cost share  $\int_{s(B_k)}^{s(B_k)+s_{i_j}}f(x)dx$, which is smaller than $\int_{s(B_j')}^{s(B_j')+s_{i_j}}f(x)dx$ since $s(B_j')\le s(B_k)-s_{t_k}<\frac34$. This  contradicts  the fact that $\pi$ is an NE packing.
\end{proof}

It is clear that $s(B_{m-1})+s(B_{m})>1$  either by $s(B_m)\ge 3/4$ or by Lemma \ref{surplus} with $s(B_m)-s_{t_m}< 3/4$. In the following, we will prove that the optimal packing uses at least $\lceil\frac{3m-3}{4}\rceil\geq\lfloor \frac{3m}{4}\rfloor$ bins, which then implies  $PoA\leq 4/3$.

\begin{theorem}\label{th:poa}
An optimal packing uses at least $\lceil\frac{3m-3}{4}\rceil$ bins, i.e., $PoA\leq \frac43$.
\end{theorem}

To prove the theorem, first note that if $s(B_{m-2})\geq \frac34$, then the total size of the $n$ items is at least $\frac34\cdot (m-2)+s(B_{m-1})+s(B_{m})>\frac{3}{4}(m-1)$. Since each bin can accommodate items with a total size at most $1$, the conclusion of the theorem follows. Hence, in the discussion below, we will always assume that $s(B_{m-2})<\frac34$.

Let $L^*:=s(B_{m-1})$. Obviously, $L^*\le s(B_{m-2})< \frac34$. Then we use $L^*$ to divide all items into three types: for any item $i\in N$, if its size $s_i\ge L^*$, we call it a \emph{super} item; if $L^*>s_i> 1-L^*$, then we call it a \emph{regular} item; otherwise $s_i\le 1-L^*$, we call it a \emph{tiny} item.
\begin{lemma}\label{notiny}
There are no tiny items in the last two bins $B_{m-1}$ and $B_m$.
\end{lemma}
\begin{proof}First we can see that the statement holds for $B_m$, because otherwise a tiny item can move into bin $B_{m-1}$, contradicting Lemma \ref{surplus}. Next, we show that $B_{m-1}$ does not contain any tiny item either. Note that if there is only one item in $B_{m-1}$, then this item must be a super item because of $s(B_{m-1})=L^*$. If there are more than one item in $B_{m-1}$,  then $s(B_{m-1})\leq s(B_{m-2})<\frac34$  implies that  $s_{t_{m-1}}>\frac14$ (by Lemma \ref{surplus}), $B_{m-1}$ consists of only two items, 
and $\frac14<s_{b_{m-1}}<\frac12$. Suppose by way of contrary that there exists a tiny item in  $B_{m-1}$.  Then  $t_{m-1}$ must be one such tiny item. By the definition of tiny item, we have $s_{t_{m-1}}\le 1-L^*$ and $s_{t_{m-1}}+s(B_m)\le s_{t_{m-1}}+L^*\le1$. Recall that $\pi$ is an NE packing, which means that $t_{m-1}$ has no incentive to move into $B_m$.  This, together with the fact that $B_m$ contains no tiny item, implies that $s_{b_{m-1}}\ge s(B_m)$.  Recalling $s_{b_{m-1}}<\frac12$, thus $s_{b_{m-1}}+ s(B_m)<1$, which means that the bottom item $b_{m-1}$ can fit into the last bin $B_m$. Since $\pi$ is an NE under the LSB cost-sharing rule with $\Lambda=\frac34$, it must be the case that $s_{t_{m-1}}\ge s(B_m)> 1-L^*$, thereby leading to a contradiction.
Now we can conclude that neither $B_{m-1}$ nor $B_m$ contains any tiny item.
\end{proof}

Recall the classification of items.  For any feasible packing, a bin can not accommodate one super item plus with a regular item. In addition, since a regular item is larger than $\frac14$, a bin can  accommodate at most $3$ of them. Now we claim that under $\pi$, except for the last bin $B_m$, each bin either contains a super item or at least two regular items. This is because, otherwise the tiny items in a violating bin would unilaterally like to move into the bin $B_{m-1}$ which contains no tiny item, contradicting the NE assumption. By a similar  argument, we arrive at the following lemma.
\begin{lemma}\label{regularsize}
Under the NE packing $\pi$, except the last bin $B_m$, the total size of all non-tiny items (i.e., the super items and regular items) in each bin is at least $L^*$.
\end{lemma}
 Assume that among the first $(m-1)$ bins of $\pi$, there are $k_1$ bins with only one super item, $k_2$ bins with exactly two regular items, and $k_3$ bins with exactly three regular items, where $k_1+k_2+k_3=m-1$. So in the first $(m-1)$ bins, there are $k_1$ super items, and $2k_2+3k_3$ regular items. Recall that there is at least one regular or super item in $B_m$. In the following, we will finally show that even to pack all the non-tiny items, every valid packing (including the optimal one)  uses at least $k_1+\lceil\frac34(k_2+k_3)\rceil$ bins. To simplify the analysis, we will just assume that $B_m$ contains only one regular item. Thus we need to pack $k_1$ super items and $2k_2+3k_3+1$ regular items. Observe that under any valid packing, a super item can not be packed with any regular item into one bin. So each super item alone will occupy one bin. So we only need to show that the $2k_2+3k_3+1$ regular items will take at least $\lceil\frac34(k_2+k_3)\rceil$ bins to pack them.
We rename all the regular items according to their size ranks such that $i_1\succ i_2\succ \cdots\succ i_{p}$, where $p=2k_2+3k_3+1$. Let $\mathcal N_r:=\{i_1,i_2,\ldots, i_{k_2+1}\}$ be the biggest $k_2+1$ regular items.
\begin{claim}\label{onlytwo}
No bin can accommodate two items in $\mathcal N_r$ with another regular item.
\end{claim}
\begin{proof}
Recall that under $\pi$, we have $k_2$ bins containing exactly two regular items (and perhaps some tiny items). Let $b_{j'}$ be the minimum bottom item of these $k_2$ bins, where its bin is denoted as $B_{j'}$. Then apparently, the size of $b_{j'}$ is no more than the size of regular item $i_{k_2}$. Therefore the total size of the two regular items in $B_{j'}$ is no more than the total size of $s(i_{k_2})+s(i_{k_2+1})$. Then by Lemma \ref{regularsize}, we know that $s(i_{k_2})+s(i_{k_2+1})\ge L^*$. It implies that any two regular items from $\mathcal N_r$ sum up to at least $L^*$. Combining with the definition of regular items, we can conclude that the total size of two regular items of $\mathcal N_r$ and  any one more regular item will exceed the bin size $1$.
\end{proof}
Finally we can prove Theorem \ref{th:poa} by showing that every valid packing (including optimal ones) uses at least $\lceil\frac{3m-3}{4}\rceil$ bins, even for packing all the non-tiny items.

\bigskip

\noindent\textbf{Proof of Theorem \ref{th:poa}}
\begin{proof}
First we show that the $p=2k_2+3k_3+1$ regular items will take at least $\lceil\frac34(k_2+k_3)\rceil$ bins to pack them. In any packing, if a bin contains 3 regular items, then among which at most one item of $\mathcal N_r$ is included. Note that there are  $p-(k_2+1)=k_2+3k_3$ regular items outside $\mathcal N_r$. Based on whether all items in $\mathcal N_r$ are able to be packed into bins with 3 regular items (one item in $\mathcal N_r$ and two regular items outside $\mathcal N_r$), the proof proceeds with the following two cases.
\begin{description}
  \item[(i)] $\frac{k_2+3k_3}{2}\geq |\mathcal N_r|=k_2+1$.  This condition is equivalent to $k_2\le 3k_3-2<3k_3$. Then under every valid packing, the number of bins that  all the $p=2k_2+3k_3+1$ regular items take is at least
  $\lceil\frac{p}{3}\rceil$. We are done by noting that \[\frac{p/3}{(k_2+k_3)}=\frac23+\frac{k_3+1}{3(k_2+k_3)}> \frac23+\frac{1}{3(\frac{k_2}{k_3}+1)}>\frac23+\frac{1}{12}=\frac34.\]

  \item[(ii)] $\frac{k_2+3k_3}{2}< |\mathcal N_r|$. In this case, we can pack at most $\frac{k_2+3k_3}{2}$ items of $\mathcal N_r$ into bins with 3 regular items. The most efficient way using the least number of bins to pack the regular items is to share one item from $\mathcal N_r$ with other two regular items out of $\mathcal N_r$, and then pack the remaining items in $\mathcal N_r$ in pairs. Suppose there are $q$ items from $\mathcal N_r$  packed into bins of 3 regular items. It is easy to see that $q\leq \frac{k_2+3k_3}{2}$. Except these $q$ regular items, all the other $k_2+1-q$ items in $\mathcal N_r$ will occupy at least $\frac{k_2+1-q}{2}$ bins. Then the total number of bins needed to pack all the regular items is at least
      \begin{eqnarray*}
      & & \frac{k_2+1-q}{2}+\frac{q+2k_2-(k_2+1)+3k_3+1}{3}\\&=&\frac{5k_2+6k_3+3}{6}-\frac q6\\&\geq& \frac{5k_2+6k_3+3}{6}-\frac{k_2+3k_3}{12}\\&=& \frac34(k_2+k_3)+\frac12,
      \end{eqnarray*} where the above inequality follows from $q\leq \frac{k_2+3k_3}{2}$.
\end{description}

 Recall that $k_1+k_2+k_3=m-1$. We complete the proof by showing that 
 \[k_1+\left\lceil\frac34(k_2+k_3)\right\rceil =\left\lceil \frac34(k_1+k_2+k_3)+\frac14 k_1 \right\rceil \geq \left\lceil\frac34(m-1)\right\rceil.\]
\end{proof}

\section{Construct Efficient Equilibria by  $BFD$}\label{se:BFD}
In Section \ref{34}, we have shown that the LSB cost-sharing rule with $\Lambda=\frac34$ is already the best possible with $PoS=1$ and $PoA=\frac43$ in the game. 
So one can not expect to reduce the inefficiency of equilibria in this sense. However, it would be favorable if a social planner can construct an NE packing in polynomial time, whose asymptotic approximation ratio is smaller than $\frac43$. In this case, the system can directly reach a stable and effective state under which no selfish item would like to move. In addition, it is desirable to compute an NE efficiently with an asymptotic approximation ratio comparable to the classic optimization algorithms for bin packing. To this end, we formally introduce a well-known bin packing algorithm $BFD$ (Best Fit Decreasing), which works as follows.

\paragraph{\bf Algorithm $BFD$.} First reorder the items in non-increasing sizes such that item $1\succ$ item $2\succ\cdots\succ$ item $n$. Then one by one in this order, the algorithm places the current item in a bin that it fits the tightest; if there is no existing bin with enough space for the item, then open a new bin and put it in.

In this section, we show that by setting $\Lambda=\frac23$ in the bin packing games, even though its $PoA$ is at least $1.5$ as illustrated in the following example, the algorithm $BFD$ always leads to  NE packings whose 
$PoA$ is $\frac{11}{9}$, as the asymptotic approximation ratio of $BFD$ is $11/9$~\cite{J74}.
\begin{example}
The item set contains $2k$ items with size $\frac23$,  $2k$ items with size $\frac13$, and $3k$ items with size $\varepsilon=\frac{1}{3k}$, where $k$ is a positive integer. In this instance, there exists a packing $\pi^*$ with $3k$ bins used as shown in Figure \ref{fig:NE} and it is an NE. To see this, first note that the items with sizes $\frac23$ and $\frac13$ can not move; besides, the tiny items with size $\varepsilon$ have no incentive to move since their costs are already zero.  Meanwhile, the optimal packing only uses $m^*=2k+1$ bins as illustrated in Figure \ref{fig:OPT}. Apparently, $$\lim_{k\rightarrow\infty}\frac{m(\pi^*)}{m^*}=1.5.$$
\end{example}
\begin{figure}[ht!]
\begin{minipage}[b]{0.4\textwidth}
\includegraphics[width=\textwidth]{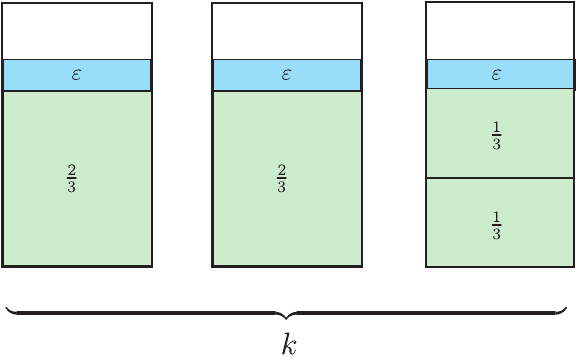}\\
\caption{$m(\pi^*)=3k$}
\label{fig:NE}
\end{minipage}
\hspace{20mm}
\begin{minipage}[b]{0.418\textwidth}
\includegraphics[width=\textwidth]{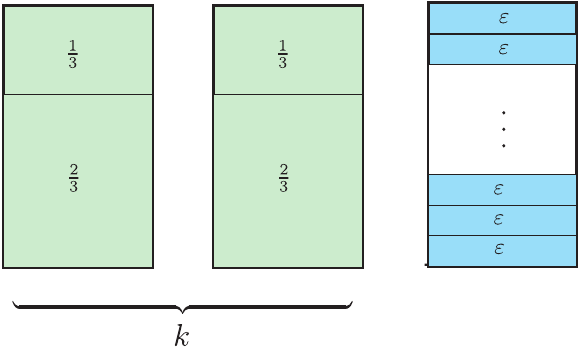}\\
\caption{$m^*=2k+1$}
\label{fig:OPT}
\end{minipage}
\end{figure}

 Denote $\pi$ as the derived packing by the algorithm $BFD$. In the following, we show that $\pi$ is not just an NE (which in fact is easy to prove), but also a strong Nash equilibrium.
\begin{definition}
A packing is called a \emph{strong Nash equilibrium} (\emph{strong NE}) if no non-empty subset of items can change their bins simultaneously such that every one of this subset benefits from this changing. 
\end{definition}
\begin{theorem}
The packing $\pi$ produced by $BFD$ is a strong NE under the LSB cost-sharing rule with $\Lambda=\frac23$.
\end{theorem}
\begin{proof}
From the process of the $BFD$ algorithm we know, that when every item $i\in N$ is considered, it is placed into a feasible bin with the highest height level. Given an arbitrary non-empty subset of items $S\subseteq N$. We show that no matter how the items in $S$ change their bins simultaneously, at least one item in $S$ can not gain by doing this. Let $j\in S$ be the biggest item of $S$.

First consider the case that $j$ is not a bottom item of $\pi$.  Since all other items which are bigger than $j$ don't change their positions and $j$ is a non-bottom item, then $j$ has no incentive to move since it has already been placed in the best bin it fits and has the minimum possible cost-share.
Now we turn to the case of $j$ being a bottom item. Suppose there are $m$ bins under packing $\pi$, and name them as $B_1,B_2,\ldots,B_m$ as their opening order under $BFD$ algorithm, i.e., $B_1$ is the first opened bin, and $B_m$ is the last one. For $k=1,2,\ldots,m$, let $b_k$ be the bottom item of bin $B_k$. Then by the $BFD$ algorithm, we have $b_1\succ b_2\succ\cdots\succ b_m$. Besides, we
should note that when each bottom item $b_k$ ($k\le m$) is considered by $BFD$, it can not be put into any existing bin $B_h$ with $h<k$. So we can see that after all items in $S$ change their bins,  item $j$ is still a bottom item in the new resulted packing (denoted by $\pi'$). Now suppose under $\pi$, $j$ is the bottom of bin $B_k$, i.e., $j=b_k$; and under the new packing $\pi'$, $j\in B'$. If $j$ can decrease its cost by the simultaneous migration of $S$ (otherwise we are done), then combining with the cost-sharing rule with $\Lambda=\frac23$, we know it must be the case that  $s(B_k)<\frac23$ (thus $B_k$ contains only one item with size at least $\frac13$) and $s(B')>s(B_k)$. Denote by $i'$ the second biggest item in bin $B'$ of $\pi'$. We consider two possible scenarios.

\begin{enumerate}
  \item[(1)] $s_{i'}>\frac13$. Under the original packing $\pi$, $i'$ must belong to some bin $B_h$ with $h<k$, because otherwise we should have $s(B_k)\ge\frac23$ since $j$ and $i'$ can be packed into one same bin by $BFD$. Recalling that bottom item $j$ can not move into bin $B_h$, we have $i'\in S$. Note that $b_h\succ j(=b_k)$. Then apparently, the cost-share of item $i'$ in $B'$ is no less than the previous cost-share in $B_h$. Thereby it has no incentive to move.
  \item[(2)] $s_{i'}\le \frac13$. Then $s_{i'}+s(B_k)\le 1$. This means, that when $i'$ was considered by the $BFD$ algorithm, it was placed in bin $B_k$ or a better bin, with a cost no more than $\int_{s_j}^{s_j+s_{i'}}f(x)dx$. So if $i'\in S$, then $i'$ can not reduce its cost by moving into the bin $B'$ of $\pi'$, done. If $i'\notin S$, then this is to tell that $j$ moves into the bin containing $i'$ and $i'$ remains in the bin. Recall that $j$ is the biggest item of $S$ and it is still a bottom item after the deviation. We know that under $\pi$, item $i'\in B_{\ell}$ with $\ell>k$. Combining with the fact that $s(B_k)<\frac23$ and $j=b_k\succ b_{\ell}$, we deduce that there are two items in $B_{\ell}$ with size larger  than $1-s_j\geq 1-s(B_k)>\frac13$, since otherwise one of the two biggest items in $B_{\ell}$ should be packed into $B_k$ by the $BFD$. Therefore under $\pi$, bin $B_{\ell}$'s size $s(B_{\ell})>\frac23$ and the cost-share of $b_{\ell}$ is exactly $\int_0^{s_{b_{\ell}}}f(x)dx$. Besides, since $j$ will move into the bin $B_{\ell}$ in the deviation of $S$ and $i'$ will be the second biggest item in the new bin, thus $b_{\ell}\in S$. Let $j'\in S$ be a bottom item in $S$ with cost-share being exactly  $\int_0^{s_{j'}}f(x)dx$ under $\pi$ and the biggest such one. Then it can be easily deduced that $s_{j'}\ge s_{b_{\ell}}>1-s_j$, and  $j'$ is still a bottom item under the new packing $\pi'$, which implies that its new cost-share under $\pi'$ is at least $\int_0^{s_{j'}}f(x)dx$. So $j'$ can not decrease its cost-share by the changes.
\end{enumerate}
Overall, we can conclude that under  packing $\pi$, no group of items can  change the bins simultaneously to make every one of them better off. It implies that $\pi$ is a strong NE.
\end{proof}

\section{Extension}\label{se:extend}
In this section, we study a variant of the classical selfish scheduling game \cite{CKN09,ILMS09,C13,K13,AFJMS15}, and show that the key idea of our cost-sharing rule can be adapted to design a best coordination mechanism for this considered scheduling game.  As far as we know, previously there was no such close connection between  cost-sharing-rule design in a resource selection game  and  scheduling-policy design in a scheduling game.

In this variant scheduling game, given sufficiently many identical machines, there are $n$ jobs, each with a processing time $p_j$ for job $j\in N=\{1,2,\ldots, n\}$. However, the machines have a fixed working time window $[d,d+T]$, where $T\ge\max_{j\in N}p_j$. To interpret this, one can imagine that the machines are controlled by crew members who have a fixed working period in a day. Each job is owned by an independent user whose goal is to minimize its own completion time by selecting a machine before the beginning time $d$. The scheduling policies on the machines thus become crucial.

As all jobs have selected their machines, we get an assignment, which is \emph{valid} if the total processing time of jobs on every machine does not exceed $T$. It is important to emphasize that a scheduling policy only works if the assignment is valid, otherwise the overloaded machines will crash. The social objective of the game is to minimize the number of machines used. We remark here that the scheduling models with the objective of minimizing the number of used machines are also an important research topic, see literature \cite{C04machine,chen2018m} for reference.

Apparently, if the social planner can control all the jobs, then from an optimization point of view, it is exactly the problem of bin packing optimization. However, from the game point of view, this selfish scheduling model (referred as \emph{selfish capacitated scheduling} in this section) is quite different from the aforementioned selfish bin packing. In this game, each machine does not distribute a fixed cost among the jobs it processes but decides the starting (or completion) times for them. And the jobs only care about their own completion times. The
existing cost-sharing rules in literature for selfish bin packing can not be applied to designing the scheduling policies.

Let us first look at what challenges we may face to guarantee efficient Nash equilibria. Note that one job may prefer to occupy a machine alone, rather than stay with other jobs. It results in a large efficiency loss. Namely, as an example, if every machine always processes jobs from the very beginning (time $d$), then the unique Nash equilibrium would be one machine for one job, which is the worst. In the following, we shall focus on coordination mechanisms under which no job would like to unilaterally move to an empty machine. Such coordination mechanisms are called \emph{reasonable}. Then the general lower bound $4/3$ on $POA$ in \cite{Dosa2019} still holds for reasonable mechanisms.
\begin{proposition}[\cite{Dosa2019}]
For selfish capacitated scheduling, any reasonable coordination mechanism has a $PoA$ at least $4/3$.
\end{proposition}

In the following, we will propose a best possible scheduling policy for this game with $PoS=1$ and $PoA \le 4/3$. The policy will be denoted by \textsf{$^{0}$ShortestFirst} and is adopted by every machine. The whole coordination mechanism first re-indexes the jobs in the non-decreasing order of their processing times (breaking ties arbitrarily) and fixes this order. Then for each machine $\mathfrak m$, suppose the valid job set it receives is $J_{\mathfrak m}$. The machine stays idle until time $s_{\mathfrak m}=d+(T-\sum_{j\in J_{\mathfrak m}}p_j)$. Then from the time $s_{\mathfrak m}$ on, the machine processes the jobs continuously in the non-decreasing order of jobs' processing times (breaking the ties as predefined above). In fact, we can view the idle times of a machine as being processing ``null jobs''. This scheduling policy can also be denoted as \textsf{LongestLast}.

It is obvious that the coordination mechanism with \textsf{$^{0}$ShortestFirst} policy is reasonable since no job favors to be completed by the deadline $d+T$. Besides, in this coordination mechanism, every job is willing to select a machine with a larger total processing time of ``longer" jobs. In this case, the job can be processed earlier. This is quite similar to the aforementioned selfish bin packing game with $\Lambda=1$, where every job prefers to be packed into a bin with a higher height level. However, in the latter game, since each bin's bottom item has to pay the remaining uncovered cost, the bottom item also prefers to stay in a bin with more content. 
Now we will try to modify the latter game to make the two games equivalent.

We consider relaxing the budget-balance constraint of $\sum_{i\in B}c(i)= 1$ to the constraint  $\sum_{i\in B}c(i)\leq 1$ in the selfish bin packing game. Then it is not hard to find that by viewing jobs as items, the normalized processing times of jobs as items' sizes, and machines as bins, then the games of selfish capacitated scheduling under \textsf{$^{0}$ShortestFirst} policy, are essentially equivalent to the games of selfish bin packing under the following cost-sharing rule, with respect to their sets of Nash equilibria and NE's inefficiency:

 \begin{itemize}
\item[1.]  Let $f(x)=2-2x$ with $x\in[0,1]$. 

\item[2.]  Given a packing $\pi$, for each bin $B$, suppose w.l.o.g. that $B=\{i_1,i_2,\ldots,i_l\}$ and $i_1\succ i_2\succ\cdots\succ i_l$. Let $S_B^j:=\sum_{h=1}^js_{i_h}$ for $j\leq l$. Then the cost-share of item $i_j$ for $j=1,2,\ldots,l$ is
\[ c_{i_j}(\pi)=\begin{cases} \int_0^{S_B^1}f(x)dx,  \  \text{  for  } j=1;\\ \int_{S_B^{j-1}}^{S_B^{j}}f(x)dx,  \ \text{  for  } j>1.
\end{cases}
\]
\end{itemize}
This bin packing game is denoted by $\mathbb{\bar G}_{bin}$. To avoid introducing more definitions and simplify the analysis, in the following we will focus on studying the bin packing game  $\mathbb{\bar G}_{bin}$ and follow the notations defined previously. The results derived below can be adapted to the selfish capacitated scheduling under the \textsf{$^{0}$ShortestFirst} policy.

\begin{observation}\label{newsamecost}
The above cost-sharing rule of game  $\mathbb{\bar G}_{bin}$ is local, monotone, fair, and polynomial-time computable, but not budget-balanced.

\end{observation}
Now we show that $\mathbb{\bar G}_{bin}$  always has an NE being an optimal packing, and its $PoA$ is the best possible.
\begin{theorem}\label{newpos}
For game $\mathbb{\bar G}_{bin}$, it's $PoS=1$.
\end{theorem}
\begin{proof}
 Assume that  the items in $N$ are $a_1\succ a_2\succ \cdots\succ a_{n}$. Let $\mathcal {P}^*$ be the set of the optimal packings which use the minimum number of bins. For two packings $\pi,\pi'\in \mathcal P^*$, we say that $\pi$ is lexicographically  higher than $\pi'$ if the height levels satisfy that, for some $\ell\le n$ we have, $l_{a_h}(\pi)=l_{a_h}(\pi')$ for every $h<\ell$ and $l_{a_{\ell}}(\pi)>l_{a_{\ell}}(\pi')$. Since $\mathcal P^*\neq\emptyset$, then there must exist a packing in $\mathcal P^*$ that is lexicographically highest, which is denoted by $\pi^*$. Note that under $\pi^*$,  no item in $N$ can increase its height level (thus reduce its cost) by unilaterally moving into another feasible bin. This is because otherwise, it contradicts the definition of $\pi^*$ as the deviation will not change the height levels of all the bigger items. Therefore $\pi^*$ is exactly an NE.
\end{proof}

\begin{theorem}\label{th:newpoa}
For game $\mathbb{\bar G}_{bin}$, it's $PoA= 4/3$.
\end{theorem}
\paragraph{Proof sketch of Theorem \ref{th:newpoa}.} The key idea of proving Theorem \ref{th:newpoa} is same as the analysis in Section \ref{poa}. However, since the cost-sharing rule of  game $\mathbb{\bar G}_{bin}$ is different from the rule in selfish bin packing with $\Lambda=\frac34$, its proof   still has some critical differences, which will be listed below.

Given an NE packing $\pi$ of game $\mathbb{\bar G}_{bin}$. Suppose without loss of generality that $\pi$ is composed of $m$ nonempty bins $B_1, B_2,\ldots, B_m$ with $s(B_1)\ge s(B_2)\ge \cdots\ge s(B_m)$. Again, for each bin $B_k$, we use $b_k$ to denote its bottom item, and $t_k$ to denote its minimum item. First, as a result of $f(x)$ being a decreasing function on the whole interval $[0,1]$, we have the following stronger conclusion (compared with Lemma \ref{surplus}) about the infeasibility of the smallest item of one bin moving into another bin.
\begin{lemma}\label{newsurplus}
For every two bins $B_j$ and $B_k$ with  $j<k$, we have $s(B_j)+s_{t_k}>1$.
\end{lemma}
Consequently,  the following corollary is easy to get.
\begin{corollary}
Except for the last bin, every other bin  $B_j$'s size $s(B_j)>1/2$, where $j\le m-1$.
\end{corollary}
Next, we  still try to prove that the optimal packing uses at least $\lfloor \frac{3m}{4}\rfloor$ bins, which implies  $PoA\leq 4/3$. However, since we can not guarantee the correctness of  Lemma \ref{notiny} in game $\mathbb{\bar G}_{bin}$, we  have a slightly different assumption base and then further arguments. Note that if $s(B_{m-3})\geq \frac34$, then the total size of the $n$ items is at least $\frac34\cdot (m-3)+s(B_{m-2})+s(B_{m-1})+s(B_{m})>\frac34\cdot (m-3)+s(B_{m-2})+1>\frac{3m}{4}-\frac{3}{4}$, which still implies  the optimal packing using at least $\lfloor \frac{3m}{4}\rfloor$ bins. So in the following, we will assume that $s(B_{m-3})<\frac34$.

Different from the definition in Section \ref{poa}, we denote $L^*:=s(B_{m-2})\le s(B_{m-3})< \frac34$. For any item $i\in N$, if its size $s_i\ge L^*$, we call it a \emph{super} item; if $L^*>s_i> 1-L^*$,  we call it a \emph{regular} item; otherwise $s_i\le 1-L^*$, we call it a \emph{tiny} item.
\begin{lemma}
There are no tiny items in the last three bins $B_{m-2}, B_{m-1}$ and $B_m$.
\end{lemma}
\begin{proof}
The statement's correctness for the last two bins $B_{m-1}, B_m$ is apparent by Lemma \ref{newsurplus}. For bin $B_{m-2}$, as a result of $s(B_{m-3})< \frac34$ and Lemma \ref{newsurplus}, we have $s_{t_{m-2}}>1/4$. Thus there are at most two items in $B_{m-2}$. If $s_{t_{m-2}}$ is not a tiny item, then we are done. Otherwise,  $s_{t_{m-2}}<s_{t_{m-1}}$ and $s_{t_{m-2}}+s(B_{m-1})\le s_{t_{m-2}}+L^*\leq 1$. Recall that $\pi$ is an NE. So item $t_{m-2}$ has no incentive to move into $B_{m-1}$. This means that the bottom item of $B_{m-2}$ satisfy $s_{b_{m-2}}\geq s(B_{m-1})>1/2$. Then $S(B_{m-2})=s_{b_{m-2}}+s_{t_{m-2}}>1/2+1/4$, contradicting the fact $s(B_{m-2})<3/4$.
\end{proof}
In the following, by focusing on all the non-tiny items in the first $(m-2)$ bins, and  using a similar analysis as the counterpart's in Section \ref{poa}, we can show that packing  all the non-tiny items in the first $(m-2)$ bins and the items in $B_{m-1}$ and $B_m$ needs at least $\lceil \frac{3m}4-\frac23\rceil\geq \lfloor\frac{3m}{4}\rfloor$ bins. \qedsymbol

\bigskip
We have  shown in Section \ref{se:analysis} that, for the game of selfish bin packing, by setting constant $\Lambda=\frac34$ in the LSB cost-sharing rule, we can guarantee its $PoS=1$ and $PoA=\frac43$. However, if we want to compute  an NE in polynomial time with $PoA$ smaller than $\frac43$, we have to reset $\Lambda=\frac23$ in the game, but with the $PoA\ge \frac32$, as illustrated in Section \ref{se:BFD}. But for game $\mathbb{\bar G}_{bin}$, we can achieve all the desired results simultaneously.

\begin{theorem}
For game $\mathbb{\bar G}_{bin}$, the packing $\pi$ produced by the algorithm $BFD$ is a strong NE.
\end{theorem}
\begin{proof}
Given an arbitrary subset of items $S\subseteq N$. Let $i\in S$ be the biggest item of $S$. We show that no matter how the items in $S$ change their bins simultaneously, item $i$ can not reduce its cost by these deviations. From the process of $BFD$ algorithm we know, every item has been placed into the best possible bin if all other bigger items don't change their positions. By recalling the cost-sharing rule of game  $\mathbb{\bar G}_{bin}$, no matter whether $i$ is a bottom item or not, its cost share can not be reduced after the deviations of $S$.
\end{proof}

\begin{remark}
Even though the game $\mathbb{\bar G}_{bin}$ is introduced as a tool to study the game of selfish capacitated scheduling in this paper, it has an independent value as a bin packing game. Note that in classical selfish bin packing, the social planner's aim is to minimize the number of bins used, not to earn revenue. So the budget-balance constraint needs not to be strict. Instead, one can  replace the constraint $\sum_{i\in B}c(i)= 1$ with $\sum_{i\in B}c(i)\leq 1$ in the game. If $\sum_{i\in B}c(i)< 1$ for a bin $B$ in some instance, the social planner can pay the remaining cost of $1-\sum_{i\in B}c(i)$.
\end{remark}

\begin{theorem}
For selfish capacitated scheduling game with \textsf{$^{0}$ShortestFirst} policy, its $PoS=1$, $PoA=\frac43$, and the algorithm $BFD$ leads to a strong NE.
\end{theorem}

\section{Conclusion}\label{se:con}
This paper proposes a best cost-sharing rule for selfish bin packing, assuming that no items in a bin are willing to afford the whole cost by moving to an empty bin. To break the lower bound of $4/3$, one has to redefine the game,  allowing negative payments or some other selfish behaviors. One open question this work still leaves is the NE existence (and also the exact $PoS$) for bin packing games with $\Gamma>\frac34$. 
We believe that the framework of our cost-sharing rule may shed light on many other variants of selfish bin packing and help design good mechanisms/protocols for more general resource selection games.



\begin{thebibliography}{10}

\bibitem{A08}
Elliot Anshelevich, Anirban Dasgupta, Jon Kleinberg, {\'E}va Tardos, Tom
  Wexler, and Tim Roughgarden.
\newblock The price of stability for network design with fair cost allocation.
\newblock {\em SIAM Journal on Computing}, 38(4):1602--1623, 2008.

\bibitem{AFJMS15}
Yossi Azar, Lisa Fleischer, Kamal Jain, Vahab Mirrokni, and Zoya Svitkina.
\newblock Optimal coordination mechanisms for unrelated machine scheduling.
\newblock {\em Operations Research}, 63(3):489--500, 2015.

\bibitem{B06}
Vittorio Bil\'o.
\newblock On the packing of selfish items.
\newblock In {\em Proceedings 20th IEEE International Parallel $\&$ Distributed
  Processing Symposium (IPDPS)}, pages 9--pp. IEEE, 2006.

\bibitem{C11}
Zhigang Cao and Xiaoguang Yang.
\newblock Selfish bin covering.
\newblock {\em Theoretical Computer Science}, 412(50):7049--7058, 2011.

\bibitem{C13}
Ioannis Caragiannis.
\newblock Efficient coordination mechanisms for unrelated machine scheduling.
\newblock {\em Algorithmica}, 66(3):512--540, 2013.

\bibitem{CRV10}
Ho-Lin Chen, Tim Roughgarden, and Gregory Valiant.
\newblock Designing network protocols for good equilibria.
\newblock {\em SIAM Journal on Computing}, 39(5):1799--1832, 2010.

\bibitem{chen2018m}
Lin Chen, Nicole Megow, and Kevin Schewior.
\newblock An $\mathcal{O}(\log {m})$-competitive algorithm for online machine
  minimization.
\newblock {\em SIAM Journal on Computing}, 47(6):2057--2077, 2018.

\bibitem{C21}
Xin Chen, Qingqin Nong, and Qizhi Fang.
\newblock An improved mechanism for selfish bin packing.
\newblock {\em Journal of Combinatorial Optimization}, 42(3):636--656, 2021.

\bibitem{CKN09}
George Christodoulou, Elias Koutsoupias, and Akash Nanavati.
\newblock Coordination mechanisms.
\newblock {\em Theoretical Computer Science}, 410(36):3327--3336, 2009.

\bibitem{C04machine}
Julia Chuzhoy, Sudipto Guha, Sanjeev Khanna, and Joseph~Seffi Naor.
\newblock Machine minimization for scheduling jobs with interval constraints.
\newblock In {\em Proceedings of the 45th annual IEEE symposium on foundations of computer science
  (FOCS)}, pages 81--90. IEEE, 2004.

\bibitem{CCGMO15}
Richard Cole, Jos{\'e}~R Correa, Vasilis Gkatzelis, Vahab Mirrokni, and Neil
  Olver.
\newblock Decentralized utilitarian mechanisms for scheduling games.
\newblock {\em Games and Economic Behavior}, 92:306--326, 2015.

\bibitem{D20}
Gy{\"o}rgy D{\'o}sa and Leah Epstein.
\newblock Quality of equilibria for selfish bin packing with cost sharing
  variants.
\newblock {\em Discrete Optimization}, 38:100556, 2020.

\bibitem{Dosa2019}
Gy{\"o}rgy D{\'o}sa, Hans Kellerer, and Zsolt Tuza.
\newblock Using weight decision for decreasing the price of anarchy in selfish
  bin packing games.
\newblock {\em European Journal of Operational Research}, 278(1):160--169,
  2019.

\bibitem{E11}
Leah Epstein and Elena Kleiman.
\newblock Selfish bin packing.
\newblock {\em Algorithmica}, 60(2):368--394, 2011.

\bibitem{E21}
Leah Epstein and Elena Kleiman.
\newblock Selfish vector packing.
\newblock {\em Algorithmica}, 83(9):2952--2988, 2021.

\bibitem{E16}
Leah Epstein, Elena Kleiman, and Juli{\'a}n Mestre.
\newblock Parametric packing of selfish items and the subset sum algorithm.
\newblock {\em Algorithmica}, 74(1):177--207, 2016.

\bibitem{de81}
W.~Vega Fernandez~de la and G.~S. Lueker.
\newblock Bin packing can be solved within $1+ \varepsilon$ in linear time.
\newblock {\em Combinatorica}, 1(4):349--355, 1981.

\bibitem{GKR16}
Vasilis Gkatzelis, Konstantinos Kollias, and Tim Roughgarden.
\newblock Optimal cost-sharing in general resource selection games.
\newblock {\em Operations Research}, 64(6):1230--1238, 2016.

\bibitem{HHSS21}
Tobias Harks, Martin Hoefer, Anja Schedel, and Manuel Surek.
\newblock Efficient black-box reductions for separable cost sharing.
\newblock {\em Mathematics of Operations Research}, 46(1):134--158, 2021.

\bibitem{Hv14}
Tobias Harks and Philipp von Falkenhausen.
\newblock Optimal cost sharing for capacitated facility location games.
\newblock {\em European Journal of Operational Research}, 239(1):187--198,
  2014.

\bibitem{H17}
Rebecca Hoberg and Thomas Rothvoss.
\newblock A logarithmic additive integrality gap for bin packing.
\newblock In {\em Proceedings of the 28th Annual ACM-SIAM Symposium on
  Discrete Algorithms (SODA)}, pages 2616--2625. SIAM, 2017.

\bibitem{ILMS09}
Nicole Immorlica, Li~Erran Li, Vahab~S Mirrokni, and Andreas~S Schulz.
\newblock Coordination mechanisms for selfish scheduling.
\newblock {\em Theoretical Computer Science}, 410(17):1589--1598, 2009.

\bibitem{J74}
D.~S. Johnson, A.~Demers, J.~D. Ullman, M.~R. Garey, and R.~L. Graham.
\newblock Worst-case performance bounds for simple one-dimensional packing
  algorithms.
\newblock {\em SIAM Journal on Computing}, 3(4):299--325, 1974.

\bibitem{K13}
Konstantinos Kollias.
\newblock Nonpreemptive coordination mechanisms for identical machines.
\newblock {\em Theory of Computing Systems}, 53(3):424--440, 2013.

\bibitem{K99}
Elias Koutsoupias and Christos Papadimitriou.
\newblock Worst-case equilibria.
\newblock In {\em Proceedings of the 16th Annual Symposium on Theoretical Aspects of Computer Science
  (STACS)}, pages 404--413. Springer, 1999.

\bibitem{M13}
Ruixin Ma, Gy{\"o}rgy D{\'o}sa, Xin Han, Hing-Fung Ting, Deshi Ye, and Yong
  Zhang.
\newblock A note on a selfish bin packing problem.
\newblock {\em Journal of Global Optimization}, 56(4):1457--1462, 2013.

\bibitem{N18}
Q.Q. Nong, T.~Sun, T.C.E. Cheng, and Q.Z. Fang.
\newblock Bin packing game with a price of anarchy of 3/2.
\newblock {\em Journal of Combinatorial Optimization}, 35(2):632--640, 2018.

\bibitem{R02}
Tim Roughgarden and {\'E}va Tardos.
\newblock How bad is selfish routing?
\newblock {\em Journal of the ACM}, 49(2):236--259, 2002.

\bibitem{vH13}
Philipp von Falkenhausen and Tobias Harks.
\newblock Optimal cost sharing for resource selection games.
\newblock {\em Mathematics of Operations Research}, 38(1):184--208, 2013.

\bibitem{Y08}
Guosong Yu and Guochuan Zhang.
\newblock Bin packing of selfish items.
\newblock In {\em Proceedings of the 4th International Workshop on Internet and Network Economics
  (WINE)}, pages 446--453. Springer, 2008.

\bibitem{Z20}
Chenhao Zhang and Guochuan Zhang.
\newblock From packing rules to cost-sharing mechanisms.
\newblock {\em Journal of Combinatorial Optimization}, pages 1--16, 2020.

\end{thebibliography}
\end{document}